\documentclass[12pt,english]{article}
\usepackage[english]{babel}
\usepackage[dvips]{graphicx}
\makeatletter
\usepackage{cite}
\usepackage{paralist}
\usepackage[cmex10]{amsmath}
\usepackage{amssymb}
\usepackage{amsthm}
\usepackage{psfrag}
\usepackage[dvips]{graphicx}
\usepackage{subfig}

\begin{document}

\def\ss{\sigma}
\def\fpzi{\frac{\partial}{\partial z_i}}
\def\fpzj{\frac{\partial}{\partial z_j}}
\def\BM{{\bf M}}
\def\BI{{\bf I}}
\def\BA{{\bf A}}
\def\BB{{\bf B}}
\def\E{{\bf E}}
\def\BD{{\bf D}}
\def\BT{{\bf T}}
\def\BX{{\bf X}}
\def\BU{{\bf U}}
\def\Bu{{\bf u}}
\def\BV{{\bf V}}
\def\Bv{{\bf v}}
\def\BI{\hbox{I}}
\def\BJ{\hbox{J}}
\def\det{\hbox{\rm det}}
\def\tr{\hbox{\rm tr}}
\def\fpx{\frac{\partial}{\partial x}}
\def\fpy{\frac{\partial}{\partial y}}
\def\det{\hbox{\rm det}}

\title{An Overview of Local Capacity in Wireless Networks}

\author{Salman Malik\footnote{HIPERCOM Project, INRIA Rocquencourt, France. Email: \texttt{salman.malik@inria.fr}} and Philippe Jacquet\footnote{Alcatel Lucent Bell Labs, Villarceaux, France. Email: \texttt{philippe.jacquet@alcatel-lucent.fr}}}

\date{}

\maketitle

\begin{abstract}

This article introduces a metric for performance evaluation of medium access schemes in wireless ad hoc networks known as local capacity. Although deriving the end-to-end capacity of wireless ad hoc networks is a difficult problem, the local capacity framework allows us to quantify the average information rate received by a receiver node randomly located in the network. In this article, the basic network model and analytical tools are first discussed and applied to a simple network to derive the local capacity of various medium access schemes. Our goal is to identify the most optimal scheme and also to see how does it compare with more practical medium access schemes. We analyzed grid pattern schemes where simultaneous transmitters are positioned in a regular grid pattern, ALOHA schemes where simultaneous transmitters are dispatched according to a uniform Poisson distribution and exclusion schemes where simultaneous transmitters are dispatched according to an exclusion rule such as node coloring and carrier sense schemes. Our analysis shows that local capacity is optimal when simultaneous transmitters are positioned in a grid pattern based on equilateral triangles and our results show that this optimal local capacity is at most double the local capacity of ALOHA based scheme. Our results also show that node coloring and carrier sense schemes approach the optimal local capacity by an almost negligible difference. At the end, we also discuss the shortcomings in our model as well as future research directions.

\end{abstract}

\section{Introduction}
\label{sec:intro}

This article presents the framework for deriving the {\em local capacity} in a wireless ad hoc network with various medium access schemes. The local capacity is defined as the average information rate received by a node randomly located in the network. The seminal work of Gupta \& Kumar~\cite{Gupta:Kumar} and the later studies, {\it e.g.}, \cite{scaling,scaling2,scaling3} quantify the capacity in wireless networks in terms of scaling laws or bounds. These results are very important but they may not provide a deeper insight into the interrelationship between the performances of various medium access schemes in the network ({\it e.g.}, the exact achievable capacity) and the possible network and protocol design issues ({\it e.g.}, trade-offs involving protocol overhead versus performance of various medium access schemes) and network parameters. Apart from being of general interest, an additional advantage of local capacity, as compared to the results on scaling laws, is that it can be derived accurately. Therefore, the local capacity framework can be used to get better insights into the designing of medium access schemes for wireless ad hoc networks.

Medium access schemes in wireless ad hoc networks can be broadly classified into two main classes: continuous time access and slotted access. In this article, our main focus is on slotted medium access although many of our results can also be applied to continuous time medium access. Within slotted medium access category, we distinguish grid pattern, node coloring, carrier sense multiple access (CSMA) and slotted ALOHA schemes. In \S \ref{sec:context}, we will briefly summarize our motivation and related works. In the following sections, we will concisely describe the analytical methods we have developed to derive the local capacity of the above mentioned medium access schemes in a wireless ad hoc network. Because these methods are developed independently, depending on the underlying design of each medium access scheme, \S \ref{sec:model} will first give a brief overview of our network model and mathematical background of our parameters of interest. In sections \ref{sec:grid_pattern}, \ref{sec:practical} and \ref{sec:aloha}, we will discuss grid pattern, node coloring, CSMA and slotted ALOHA schemes respectively. First we will give a brief overview of the protocol overheads associated with each scheme and later, in each case, we will discuss the methods we used to evaluate local capacity of wireless ad hoc networks. In \S \ref{sec:simulations}, we will discuss the evaluation of these schemes and our results. The shortcomings and future extensions to our work and concluding remarks can be found in sections \ref{sec:future} and \ref{sec:conclude} respectively.

\section{Related Works}
\label{sec:context}

In one of the first analyses on capacity of medium access schemes in wireless networks, \cite{Nelson:Kleinrock} studied slotted ALOHA and despite using a very simple geometric propagation model, the result is similar to what can be obtained under a realistic SIR based interference model (non-fading, SIR threshold of $10.0$ and attenuation coefficient of $4.0$). Under a similar propagation model and assuming that all nodes are within range of each other, \cite{CSMA} evaluated CSMA scheme and compared it with slotted ALOHA in terms of throughput. \cite{Bartek} used simulations to analyze CSMA under a realistic SIR based interference model and compared it with ALOHA (slotted and un-slotted). For simulations, \cite{Bartek} assumed Poisson distributed transmitters with density $\lambda$. Each transmitter sends packets to an assigned receiver located at a fixed distance of $a\sqrt{\lambda}$, for some $a>0$.  

\cite{Weber,Weber2} studied the {\em transmission capacity}, which is defined as the maximum number of successful transmissions per unit area at a specified outage probability, of ALOHA and code division multiple access (CDMA) medium access schemes. They assumed that simultaneous transmitters form a homogeneous Poisson point process (PPP) and used the same model for locations of receivers as in~\cite{Bartek}. The fact that the receivers are not a part of the network (node distribution) model and are located at a fixed distance from the transmitters is a simplification. An accurate model of wireless networks should consider that the transmitters, transmit to receivers which are randomly located in their neighborhood or, in other words, their reception areas. Although, \cite{Weber} analyzed the case where receivers are located at a random distance but it was assumed that a transmitter employs transmit power control such that signal power at its intended receiver is some fixed value. 

PPP only accurately models an ALOHA based scheme where transmitters are independently and uniformly distributed over the network area. However, under exclusion schemes, like node coloring or CSMA, modeling simultaneous transmitters as PPP leads to an inaccurate representation of interference distribution. On the other hand, the correlation between the location of simultaneous transmitters makes it extremely difficult to develop a tractable analytical model and derive the closed-form expression for the probability distribution of interference. Some of the proposed approaches are as follows. \cite{CSMA-Model} (Chapter $18$) used a Mat\'ern point process for CSMA based schemes whereas \cite{Busson} proposed to use Simple Sequential Inhibition ({\em SSI}) or an extension of {\em SSI} called {\em SSI$_k$} point process. In \cite{Guard,Guard2}, interferers are modeled as PPP and outage probability is obtained by excluding or suppressing some of the interferers in the guard zone around a receiver. \cite{Weber3} analyzed transmission capacity in networks with general node distributions under a restrictive hypothesis that density of interferers is very low and, asymptotically, approaches zero. They derived bounds of high-SIR transmission capacity with ALOHA, using PPP, and CSMA using Mat\'ern point process. 

Other related works include the analysis of local (single-hop) throughput and capacity with slotted ALOHA in networks with random and deterministic node placement, and TDMA, in $1D$ line-networks only \cite{Haenggi}. \cite{Zorzi2} determined the optimum transmission range under the assumption that interferers are distributed according to PPP whereas \cite{SR-ALOHA} gave a detailed analysis on the optimal probability of transmission for ALOHA which optimizes the product of simultaneously successful transmissions per unit of space by the average range of each transmission. \cite{unslotted_csma} develops an analytical model for analysis of saturation throughput of un-slotted CSMA with collision avoidance (CSMA/CA) in wireless networks. The numerical results of the analytical model are also verified using an event driven simulator. \cite{ieee80211_mac} uses queueing analysis to perform a comprehensive throughput and delay analysis of IEEE 802.11 MAC protocol. 

\section{System Model}
\label{sec:model}

We present our mathematical model in \S \ref{model_assumptions} and discuss the connection of local capacity with transport capacity in \S \ref{sec:transport}.

\subsection{Mathematical Model and Assumptions}
\label{model_assumptions}

We consider a wireless ad hoc network where an infinite number of nodes are uniformly distributed over an infinite $2D$ map. In slotted medium access, at any given slot, simultaneous transmitters in the network are distributed like a set of points 
$$
{\cal S}=\{z_1,z_2,\ldots,z_n,\ldots\}~,
$$ 
where $z_i$ is the location of transmitter $i$. The spatial distribution of simultaneous transmitters, {\it i.e.} the set ${\cal S}$, depends on the medium access scheme employed by the nodes. Therefore, we do not adopt any {\em universal} model for the locations of simultaneous transmitters and only assume that, in all slots, the set ${\cal S}$ has a homogeneous density equal to $\lambda$. 

We consider that each transmitter employs unit transmit power. The channel gain from node $i$ to noise $j$ is represented by $\gamma_{ij}$ such that the received power at node $j$ is $P_i\gamma_{ij}$, where $P_i$ is the transmit power of node $i$. We ignore multi-path fading or shadowing effects and assume that the channel gain is solely determined by the distance and attenuation coefficient. Therefore, 
$$
\gamma_{ij}=\frac{1}{\vert z_i-z_j\vert^{\alpha}}~,
$$ 
where $\vert .\vert$ is the Euclidean norm of the vector and \mbox{$\alpha>2$} is the attenuation coefficient. Therefore, the transmission from node $i$ to node $j$ is successful only if the following condition is satisfied
\begin{equation}
\frac{\vert  z_i-z_j\vert^{-\alpha}}{N_0+\sum_{k\neq i}\vert  z_k-z_j\vert^{-\alpha}}\geq \beta~,
\label{eq:sinr_condition}
\end{equation}
where $N_0$ is the background noise (ambient/thermal) power and $\beta$ is the minimum signal to interference ratio (SIR) threshold required for successfully receiving the packet. We assume that the background noise power, $N_0$, is negligible. 
Therefore, the SIR of transmitter $i$ at any point $z$ on the plane is given by
\begin{equation}
S_i(z)=\frac{\vert  z-z_i\vert^{-\alpha}}{\sum_{j\neq i}\vert  z-z_j\vert^{-\alpha}}~.
\label{eq:sinr}
\end{equation} 

We call the reception area of transmitter $i$, the area of the plane, ${\cal A}(z_i,\lambda,\beta,\alpha)$, where this transmitter is received with SIR at least equal to $K$. The area ${\cal A}(z_i,\lambda,\beta,\alpha)$ also contains the point $z_{i}$ since here the SIR is infinite. The average size of ${\cal A}(z_i,\lambda,\beta,\alpha)$ is $\sigma(\lambda,\beta,\alpha)$, {\it i.e.}, 
$$
\sigma(\lambda,\beta,\alpha)=\E(|{\cal A}(z_i,\lambda,\beta,\alpha)|)~,$$ 
where $|{\cal A}|$ is the size of an area ${\cal A}$. Note that $\sigma(\lambda,\beta,\alpha)$ does not depend on the location of transmitter $i$, $z_i$.

Our principal performance metric is local capacity, hereafter referred to as capacity only, which is defined as the average information rate received by a receiver {\em randomly} located in the network. Consider a receiver at a random location $z$ in the network and let $N(z,\beta,\alpha)$ denote the number of reception areas it belongs to. Under general settings, following identity has been proved in \cite{Jacquet:2009}
\begin{equation}
\E(N(z,\beta,\alpha))=\lambda\sigma(\lambda,\beta,\alpha)~.
\label{eq:avg_no}
\end{equation}
$\E(N(z,\beta,\alpha))$ represents the average number of transmitters from which a receiver, randomly located in the network, can receive with SIR at least equal to $\beta$. Under the hypothesis that a node can only receive at most one packet at a time, {\it e.g.}, when \mbox{$\beta>1$}, then \mbox{$N(z,\beta,\alpha)\le 1$}. The average information rate received by the receiver, $c(z,\beta,\alpha)$, is therefore equal to $\E(N(z,\beta,\alpha))$ multiplied by the nominal capacity. Without loss of generality, we assume unit nominal capacity and we will derive 
\begin{equation}
c(z,\beta,\alpha)=\E(N(z,\beta,\alpha))=\lambda\sigma(\lambda,\beta,\alpha)~. 
\label{eq:poisson_hand_over_no}
\end{equation}

We will derive the capacity in wireless ad hoc networks with grid pattern, node coloring, CSMA and ALOHA medium access schemes. We will also show that maximum capacity can be achieved with grid pattern schemes. Wireless networks of grid topologies are studied in, {\it e.g.}, \cite{Liu:Haenggi,Hong:Hua} and compared to networks with randomly distributed nodes. In contrast to these works, we assume that only the simultaneous transmitters form a regular grid pattern. 

\subsection{Relationship of Local Capacity with Transport Capacity}
\label{sec:transport}

Gupta \& Kumar~\cite{Gupta:Kumar} introduced the concept of {\em transport capacity} which measures the end-to-end sum throughput of the network multiplied by the end-to-end distance. A important aspect of their work and the following works on transport capacity is that it is not possible to compute the exact transport capacity in terms of network and system parameters and most of the results are in the form of bounds and scaling laws, {\it i.e.}, the density of transport capacity scales as $Ck_1\sqrt{\lambda}$ bit-meters per second per unit area where $C$ is the nominal capacity and \mbox{$k_1>0$} depends on the medium access scheme and system parameters. We consider homogeneous traffic distribution in the network, {\it i.e.}, all nodes have equal priority and generate traffic at the same rate. Therefore, if all nodes are capable of transmitting at $C$ bits per second, the capacity of each node is $Ck_1/\sqrt{\lambda}$ bit-meters per second. 

It is also shown in~\cite{Gupta:Kumar} that under general settings, the effective radius of transmission is $k_2/\sqrt{\lambda}$ for some \mbox{$k_2>0$} which also depends on the medium access scheme and system parameters. If each node transmits to a receiver which is randomly located within its effective radius of transmission or, in other words, its reception area, the information rate received by a receiver is constant and equal to $Ck_1/k_2$ bits per second. In this article, we evaluate the {\em average} of this information rate received by a receiver randomly located in the network, {\it i.e.}, the local capacity. Note that, this capacity also incorporates the pre-constants associated with the scaling law, {\it e.g.} $k_1$ and $k_2$, and is independent of $\lambda$ as it is invariant for any {\em homothetic transformation} of the set of transmitters. Note that, we have abstracted out the multi-hop aspect of wireless ad hoc networks and this allows us to focus on the {\em localized} capacity. Our model also captures a realistic scenario of wireless ad hoc networks where each transmitter communicates with a receiver which is randomly located in its neighborhood or reception.

\section{Grid Pattern Based Schemes}
\label{sec:grid_pattern}

It can be argued that optimal capacity in wireless ad hoc networks can be achieved if simultaneous transmitters are positioned in a grid pattern. However, designing a medium access scheme, which ensures that simultaneous transmitters are positioned in a grid pattern, is very difficult because of the limitations introduced by wave propagation characteristics and actual node distribution. For this, location aware nodes may be useful but the specification of a distributed medium access scheme that would allow grid pattern transmissions is beyond the scope of this article. 

In this section, first we will present the analysis we used to investigate the optimality of grid patterns based medium access schemes. Later, we will also discuss the analytical method we used to to analyze the capacity of such schemes. Grid pattern based medium access schemes may have no practical implementation but their evaluation is interesting in order to establish an upper bound on the optimal capacity in wireless ad hoc networks. In the later sections, we will also discuss more practical medium access schemes.  

\subsection{Optimality of Grid Pattern Based Schemes}
\label{sec:grid_optimality}

In this section also, we consider that an infinite number of transmitters are uniformly distributed like a set of points 
$$
{\cal S}=\{z_1,z_2,\ldots,z_n,\ldots\}~,
$$
on an infinite $2D$ plane. The location of transmitter $i$ is denoted by $z_i$ and the center of the plane is at $(0,0)$. 

In order to simplify our analysis, we define a function $s_i(z)$ as 
$$
s_i(z)=\frac{|z-z_i|^{-\alpha}}{\sum_{j}|z-z_j|^{-\alpha}}~,
$$ 
where \mbox{$\alpha>2$}. The function $s_i(z)$ is similar to the SIR function $S_i(z)$, in (\ref{eq:sinr}), except that the summation in the denominator factor also includes the numerator factor. In order to simplify the notations, we will remove the reference to $z$ when no ambiguity is possible. 

We also define a function $f(s_i)$ which can be continuous or integrable. For instance, here we will use 
$$
f(s_i)=1_{s_i(z)\geq \beta'}~,
$$ 
for some given $\beta'$. Note that, in this case, the function is not continuous but we will not bother with this. In the following discussion, we can consider without loss of generality that the value of $\beta'$ is related to $\beta$ by 
$$
\beta'=\frac{\beta}{\beta+1}~.
$$ 
Therefore, if transmitter $i$ is received successfully at location $z$ (or in other words, with SIR at least equal to $\beta$, {\it i.e.}, \mbox{$S_i(z)\geq \beta$}), then \mbox{$s_i(z)\geq \beta'$} and $f(s_i)$ is equal to $1$.

We also assume a virtual disk on the plane centered at $(0,0)$ and of radius $R$. This allows us to express the density of set ${\cal S}$, $\nu({\cal S})$, in terms of the number of transmitters covered by the disk of radius $R$ or area $\pi R^2$, when $R$ approaches infinity, and it is given by a limit as
$$
\nu({\cal S})=\lim_{R\to\infty}\frac{1}{\pi R^2}\sum_i 1_{|z_i|\le R}~.
$$

We denote \mbox{$h(z)=\sum_i{f(s_i)}$}. Note that, $h(z)$ is equal to the number of transmitters which can be successfully received at $z$ and its maximum value shall be $1$ if \mbox{$\beta>1$}. 

We define $\E(h(z))$ by the limit
$$
\E(h(z))=\lim_{R\to\infty}\frac{1}{\pi R^2}\int_{|z|\le R}h(z)dz^2~.
$$
Note that, the integration is over an infinite plane or, in other words, over the disk of radius $R$ where $R$ approaches infinity. Also note that the notations are simplified by taking $dxdy$ equal to $dz^2$. We denote the reception area of an arbitrary transmitter $i$ as 
$$
\ss_i=\int f(s_i)dz^2~,
$$ 
and we have 
$$
\E(h(z))=\lim_{R\to\infty}\frac{1}{\pi R^2}\sum_i 1_{|z_i|\le R}\ss_i=\nu({\cal S})\E(\ss_i)~,
$$
with
$$
\E(\ss_i)=\lim_{n\to\infty}\frac{1}{n}\sum_{i\le n}\ss_i~.
$$
As $R$ approaches infinity, $n$, {\it i.e.}, the number of transmitters in the set ${\cal S}$, covered by the disk of radius $R$, also approaches infinity. 

Our objective is to optimize $\E(h(z))$ whose definition is equivalent to the definition of $\E(N(z,K,\alpha))$ and therefore capacity, $c(z,K,\alpha)$, as well in expressions (\ref{eq:avg_no}) and (\ref{eq:poisson_hand_over_no}) respectively. 

\newtheorem{theorem}{Theorem}
\newtheorem{lemma}{Lemma}

\subsubsection{First Order Differentiation}
\label{sec:grid_first_order}

We denote the operator of differentiation w.r.t. $z_i$ by $\nabla_i$. For \mbox{$i\neq j$}, we have 
$$
\nabla_i s_j=\alpha s_is_j\frac{z-z_i}{|z-z_i|^2}~,
$$ 
and 
$$
\nabla_i s_i=\alpha(s_i^2-s_i)\frac{z-z_i}{|z-z_i|^2}~.
$$ 
Therefore,
\begin{align}
\nabla_i h(z)&=\nabla_i\sum_if(s_i) \notag \\
&=f'(s_i)\nabla_is_i+\sum_{j\neq i}f'(s_j)\nabla_is_j \notag \\
&=\alpha s_i\frac{z-z_i}{|z-z_i|^2}\Big(-f'(s_i)+\sum_j s_jf'(s_j)\Big)~. \notag
\end{align}
Although, we know that 
$$
\int h(z)dz^2=\infty~,
$$ 
we nevertheless have a finite $\nabla_i \int h(z)dz^2$. In other words, the sum $\sum_{j}\nabla_i\ss_j$ converges for all $i$. 

\begin{lemma}
For all $j$ in ${\cal S}$, 
$$
\sum_{i}\nabla_i\ss_j=0~.$$ 
Indeed this would be the differentiation of $\ss_j$ when all points in ${\cal S}$ are translated by the same vector. Similarly, 
$$
\sum_{i}\nabla_i\int  h(z)dz^2=0~.
$$
\end{lemma}

\begin{theorem}
If the points in the set ${\cal S}$ are arranged in a grid pattern then 
$$
\nabla_i \int h(z)dz^2=\sum_{j}\nabla_i\ss_j=0~,
$$ 
and grids patterns are {\em locally} optimal. 
\end{theorem}

\begin{proof}
If ${\cal S}$ is a set of points arranged in a grid pattern, then 
$$
\nabla_i \int h(z)dz^2=\sum_{j}\nabla_i\ss_j~,
$$ 
would be identical for all $i$ and, therefore, would be null since 
$$
\sum_{i}\nabla_i\int  h(z)dz^2=0~.
$$ 

We could erroneously conclude that,
\begin{compactitem}[-]
\item all grid sets are optimal and
\item all grid sets give the same $\E(h(z))$.
\end{compactitem}
In fact this is wrong. We could also conclude that $\E(\ss_i)$ does not vary but this will contradict that $\nu({\cal S})$ {\em must} vary. The reason of this error is that a grid set cannot be modified into another grid set with a {\em uniformly bounded transformation}, unless the two grid sets are just simply translated by a simple vector. 
\end{proof}

However, we have proved that the grid sets are locally optimal within sets that can be uniformly transformed between each other. In order to cope with uniform transformation and to be able to transform a grid set to another grid set, we will introduce the linear group transformation. 

\subsubsection{Linear Group Transformation}

Here, we assume that the points in the plane are modified according to a continuous linear transform $M(t)$ where ${\bf M(t)}$ is a matrix with \mbox{${\bf M(0)}=\BI$}, {\it e.g.}, \mbox{${\bf M(t)}=\BI+{\bf t A}$} where $\BA$ is a matrix. 

Without loss of generality, we only consider $\ss_0$, {\it i.e.}, the reception area of the transmitter at $z_0$ which can be located anywhere on the plane. Under these assumptions, we have
$$
\frac{\partial}{\partial t}\ss_0=\sum_i (\BA z_i.\nabla_i\ss_0)=\tr\Big(\sum_{i} \BA^Tz_i\otimes \nabla_i\ss_0\Big)~.
$$

In other words, using the identity 
$$
\frac{\partial \tr(\BA^T\BB)}{\partial \BA}=\BB~,
$$ 
the derivative of $\ss_0$ w.r.t. matrix $\BA$ is exactly equal to 
$$\BD=\sum_{i} z_i\otimes \nabla_i\ss_0~,$$ 
such that
$$
\BD=\left[
\begin{array}{cc}
D_{xx}&D_{xy}\\
D_{yx}&D_{yy}
\end{array}
\right]~.
$$ 

Therefore, we can write the following identity
$$
\tr\Big(\BA^T \frac{\partial}{\partial \BA}\ss_0\Big)=\frac{\partial}{\partial t}\ss_0(t,\BA)\Big|_{t=0}~,
$$
where $\ss_0(t,\BA)$ is the transformation of $\ss_0$ under $M(t)$, {\it i.e.}, 
$$
\ss_0(t,\BA)=\det(\BI+{\bf t A})\ss_0~.
$$ 
We assume that {${\bf M(t)}={\bf (1+t)I}$} with \mbox{$\BA=\BI$}, {\it i.e.}, the linear transform is homothetic.

\begin{theorem}
$\BD$ is symmetric and $\tr(\BD)=2\ss_0$.
\end{theorem}

\begin{proof}
Under the given transform, 
$$
\ss_0(t,\BA)=\ss_0(t,\BI)=(1+t)^2\ss_0~.
$$ 
As a first property, we have 
$$
\tr(\BD)=2\ss_0~,
$$ 
since the derivative of $\ss_0$ w.r.t. identity matrix $\BI$ is exactly $2\sigma_0$ ({\it i.e.}, \mbox{$\tr(\BA^T\BD)=\tr(\BD)=\sigma_0'(0,\BI)=2\ss_0$}). The second property that $\BD$ is a symmetric matrix is not obvious. The easiest proof of this property is to consider the derivative of $\ss_0$ w.r.t. the rotation matrix $\BJ$ given by
$$
\BJ=\left[
\begin{array}{cc}
0&-1\\
1&0
\end{array}
\right]~,
$$
which is zero since $\BJ$ is the initial derivative for a rotation and reception area is invariant by rotation. 
Therefore, 
$$
\tr(\BJ^T\BD)=D_{yx}-D_{xy}=0~,
$$ 
which implies that $\BD$ is symmetric. 
\end{proof}

Note that $\BD$ can also be written in the following form
$$
\BD=\sum\limits_{i}z_i\otimes\nabla_i\ss_0=\int dz^{2}\sum\limits_{i}z_i\otimes\nabla_if(s_0)~.
$$
Let $\BT$ be defined as
$$
\BT=\int dz^2\sum_i (z-z_i)\otimes\nabla_i f(s_0)~,
$$
such that
$$
\BD=\int \sum_{i} z\otimes \nabla_i f(s_0)dz^2 - \BT~.
$$
The purpose of these definitions will become evident from theorems $3$ and $4$.

\begin{theorem}
We will show that $\int \sum_{i} z\otimes \nabla_i f(s_0)dz^2$ is equal to $\ss_0\BI$ and, therefore, 
$$
\BD=\ss_0\BI-\BT~.
$$ 
We will also prove that $\BT$ is symmetric. 
\end{theorem}

\begin{proof}
From the definition of $\BT$, we can see that the sum \mbox{$\sum_i (z-z_i)\otimes \nabla_i f(s_0)$} leads to a symmetric matrix since
\begin{eqnarray*}
\BT&=&\alpha\int f'(s_0)\Big(\frac{s_0^2-s_0}{|z-z_0|^2}(z-z_0)\otimes(z-z_0)+\\
&&\sum_{i\neq 0}\frac{s_0s_i}{|z-z_i|^2}(z-z_i)\otimes(z-z_i)\Big)dz^2~,
\end{eqnarray*}
and the left hand side is made of \mbox{$(z-z_i)\otimes(z-z_i)$} which are symmetric matrices. This implies that $\BT$ is also symmetric.

We can see that 
$$
\sum_{i}\nabla_i f(s_0)=-\nabla f(s_0)~,
$$ 
and using integration by parts we have
\begin{eqnarray*}
\lefteqn{\int \sum_i z\otimes \nabla_i f(s_0)dz^2=-1\times}\\
&\Biggl[
\begin{array}{cc}
\int x \fpx f(s_0)dxdy&\int x \fpy f(s_0)dxdy\\
\int y \fpx f(s_0)dxdy&\int y \fpy f(s_0)dxdy
\end{array}
\Biggr]=
\Biggl[\begin{array}{cc}
\ss_0&0\\
0&\ss_0
\end{array}\Biggr]~,
\end{eqnarray*}
which is symmetric and equal to $\ss_0\BI$. The sum/difference of symmetric matrices is also a symmetric matrix and, therefore, $\BD$ is a symmetric matrix and $\BD=\ss_0\BI-\BT$. 
\end{proof}

Now, we will only consider grid patterns and, by virtue of a grid pattern, we can have 
$$
\E(\ss_i)=\ss_0=\int f(s_0)dz^2~,
$$ 
and \mbox{$\E(h(z))=\nu({\cal S})\ss_0$}. Under homothetic transformation, $\nu({\cal S})$ and $\ss_0$ are transformed but $\nu({\cal S})\ss_0$ remains invariant.

\begin{theorem}
If the pattern of the points in set ${\cal S}$ is optimal w.r.t. linear transformation of the set, $\BD=\ss_0\BI$ and $\BT=0$.
\end{theorem}

\begin{proof}
The derivative of $\ss_0$ w.r.t. matrix $\BA$ is exactly equal to $\BD$. 
Similarly, under the same transformation
$$
\frac{\partial}{\partial t}\nu({\cal S})=\frac{1}{\det(\BI+\BA t)}\nu({\cal S})~,
$$
and for $\BA=\BI$, it can be written as 
$$
\nu'({\cal S})(t,\BI)=\nu({\cal S})/(1+t)^2~.
$$

In any case, the derivative of $\nu({\cal S})$ w.r.t. matrix $\BA$ is exactly equal to $-\BI\nu({\cal S})$. We also know that if the pattern is optimal w.r.t. linear transformation, the derivative of $\nu({\cal S})\ss_0$ w.r.t. to matrix $\BA$ shall be null. This implies that
$$
\nu({\cal S})\BD-\BI\nu({\cal S})\ss_0=0~,
$$
which leads to $\BD=\ss_0\BI$ and $\BT=0$.
\end{proof}

\begin{figure*}[!t]
\centering
\psfrag{a}{$\sqrt{3}d$}
\psfrag{b}{$2d$}
\psfrag{c}{$d$}
\includegraphics[scale=0.5]{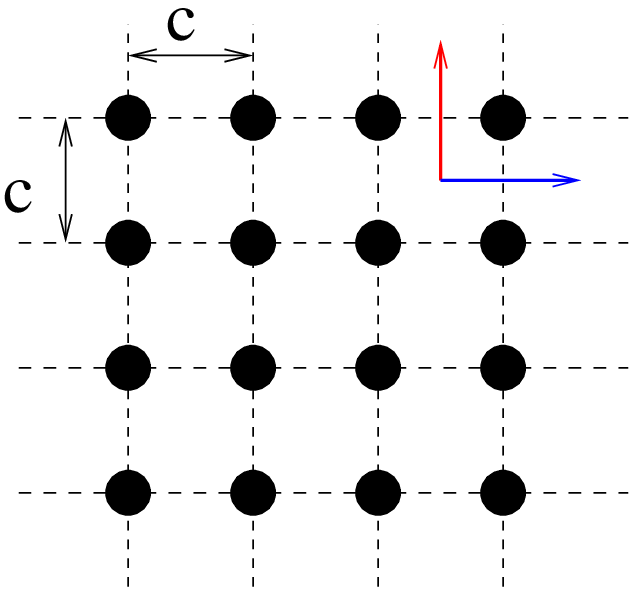}
\includegraphics[scale=0.5]{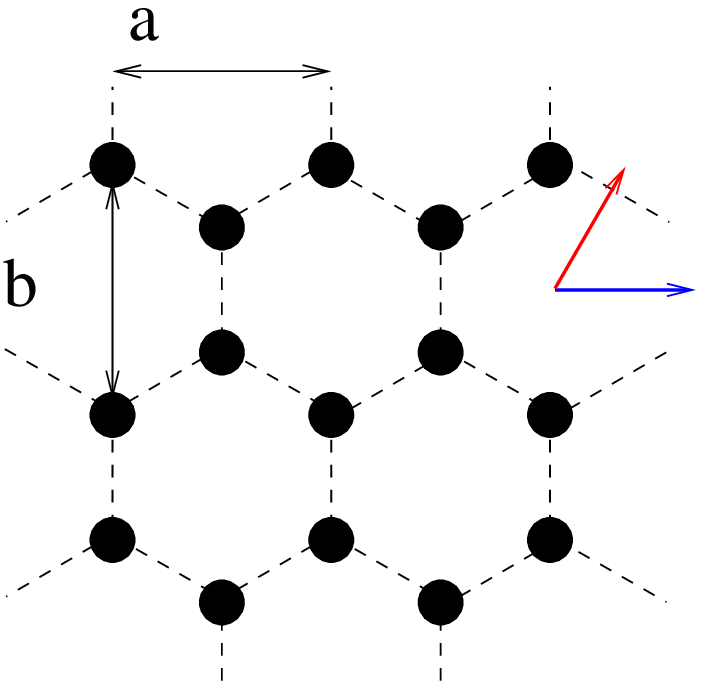}
\includegraphics[scale=0.5]{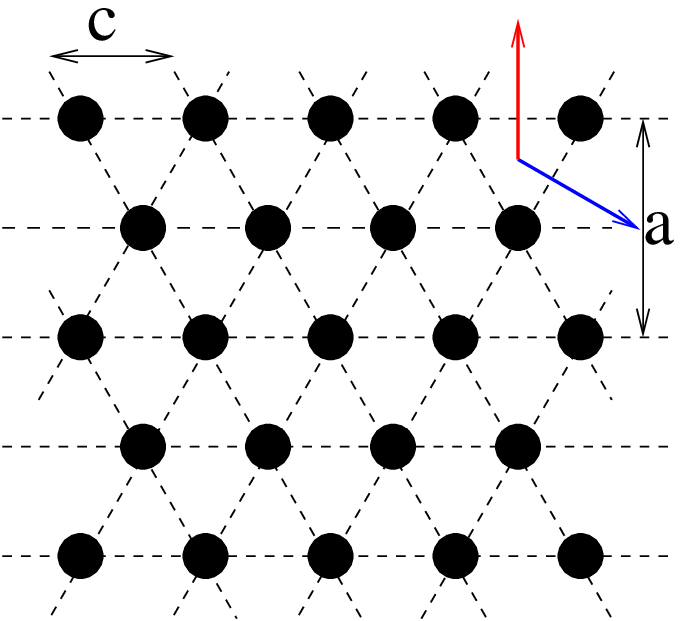}
\caption{Square, Hexagonal and Triangular grid patterns. The arrows (blue and red) represent the invariance of Eigen values w.r.t. isometric symmetries of the grids.}
\label{fig:grid_layouts}
\end{figure*}

We know that $\BT$ is symmetric and \mbox{$\BT=0$}. Thus, \mbox{$\tr(\BT)=0$}, {\it i.e.}, Eigen values are invariant by rotation. When a grid is optimal, we must have \mbox{$\BT=0$}. In any case, the matrix $\BT$ must be invariant w.r.t. isometric symmetries of the grid. On $2D$ plane, the grid patterns which satisfy this condition are square, hexagonal and triangular grids. The square grid is symmetric w.r.t. any horizontal or vertical axes of the grid and, in particular, with rotation of $\pi/2$ represented by $\BJ$. Therefore, the {\em Eigen system} must be invariant by rotation of $\pi/2$. This implies that the {\em Eigen values} are the same and therefore null since \mbox{$\tr(\BT)=0$}. Same argument also applies for the hexagonal grid with the invariance for $\pi/3$ rotation and for the triangular pattern with invariance for $2\pi/3$ rotation. 

\subsection{Reception Areas}
\label{sec:rx_area_2}

\begin{figure}[!t]
\centering
\psfrag{a}{$z_i$}
\psfrag{b}{$z$}
\psfrag{c}{$dz=J\frac{\nabla S_{i}(z)}{|\nabla S_{i}(z)|}\delta t$}
\psfrag{d}{${\cal C}(z_i,\beta,\alpha)$}
\psfrag{e}{$A(z_i,\lambda,\beta,\alpha)$}
\includegraphics[scale=0.65]{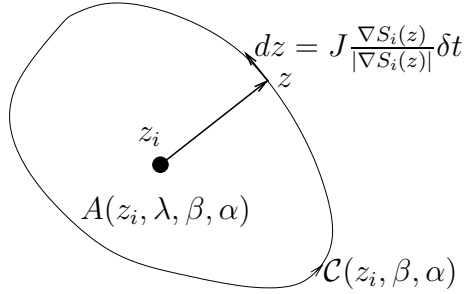}
\caption{Computation of the reception area of transmitter $i$.\label{fig:snr_gradient}}
\end{figure}

The simultaneous transmitters, {\it i.e.}, the set ${\cal S}$ is a set of points arranged in a grid pattern. We consider that, for every slot, the grid pattern is the same {\em modulo} a translation. We have covered grid layouts of square, hexagonal and triangle as shown in Fig. \ref{fig:grid_layouts}. Grids are constructed from $d$ which defines the minimum distance between neighboring transmitters and can be derived from the hop-distance parameter of a typical TDMA-based scheme. The density of grid points, $\lambda$, depends on $d$. However, the capacity, $c(z,\beta,\alpha)$, is independent of the value of $d$ or, for that matter, $\lambda$ as it is invariant for any homothetic transformation of the set of transmitters. 

Our aim is to compute the size of the reception area, $A(z_i,\lambda,\beta,\alpha)$, around each transmitter $i$. By consequence of the regular grid pattern, all reception areas are the same {\em modulo} a translation (and a rotation for the hexagonal pattern), and their surface area size, $\sigma(\lambda,\beta,\alpha)$, is the same.  

If ${\cal C}(z_i,\beta,\alpha)$ is the closed curve that forms the boundary of $A(z_i,\lambda,\beta,\alpha)$ and $z$ is a point on ${\cal C}(z_i,\beta,\alpha)$, we have
\begin{equation}
\sigma(\lambda, \beta,\alpha)=\frac{1}{2}\displaystyle\int\limits_{{\cal C}(z_i,\beta,\alpha)}\det(z-z_{i},dz)~,
\label{eq:area_integral}
\end{equation}
where $\det(a,b)$ is the determinant of vectors $a$ and $b$ and $dz$ is the vector tangent to ${\cal C}(z_i,\beta,\alpha)$ at point $z$. \mbox{$\det(z-z_i,dz)$} is the cross product of vectors $(z-z_{i})$ and $dz$ and gives the area of the parallelogram formed by these two vectors. Note that, \eqref{eq:area_integral} remains true if $z_{i}$ is replaced by any interior point of $A(z_i,\lambda,\beta,\alpha)$. 

The SIR $S_{i}(z)$ of transmitter $i$ at point $z$ is given by (\ref{eq:sinr}). We assume that at point $z$, $S_{i}(z)=\beta$. On point $z$ we can also define the gradient of $S_{i}(z)$, $\nabla S_{i}(z)$,
$$
\nabla S_{i}(z)=\left[\begin{array}{c}
				\frac{\partial}{\partial x}S_{i}(z)\\
				\frac{\partial}{\partial y}S_{i}(z)\end{array}\right]~.
$$
Note that $\nabla S_i(z)$ is inward normal to the curve ${\cal C}(z_i,\beta,\alpha)$ and points towards $z_i$. The vector $dz$ is co-linear with $J\frac{\nabla S_{i}(z)}{|\nabla S_{i}(z)|}$ where $J$ is the anti-clockwise rotation of $3\pi/2$ (or clockwise rotation of $\pi/2$) given by
$$
J=\left[\begin{array}{cc}
0 & 1\\
-1 & 0\end{array}\right]~.
$$ 
Therefore, we can fix $dz=J\frac{\nabla S_{i}(z)}{|\nabla S_{i}(z)|}\delta t$ and in (\ref{eq:area_integral}) 
\begin{align*}
\det(z-z_{i},dz)&=(z-z_i)\times J\frac{\nabla S_{i}(z)}{|\nabla S_{i}(z)|}\delta t\\
&=-(z-z_{i}).\frac{\nabla S_{i}(z)}{|\nabla S_{i}(z)|}\delta t~,
\end{align*}
where $\delta t$ is assumed to be infinitesimally small. The sequence of points $z(k)$ computed as 
\begin{align*}
z(0) & =z\\
z(k+1) & =z(k)+J\frac{\nabla S_{i}(z(k))}{|\nabla S_{i}(z(k))|}\delta t~,\end{align*} gives a discretized and numerically convergent parametric representation of ${\cal C}(z_i,\beta,\alpha)$ by finite elements. 

Therefore, (\ref{eq:area_integral}) reduces to
\begin{equation}
\sigma(\lambda, \beta,\alpha)\approx-\frac{1}{2}\sum_{k}(z(k)-z_{i}).\frac{\nabla S_{i}(z(k))}{|\nabla S_{i}(z(k))|}\delta t~,
\label{eq:area_integral_2}
\end{equation}
assuming that we stop the sequence $z(k)$ when it loops back on or close to the point $z$. 

The point, \mbox{$z(0)=z$}, can be found using Newton's method. First approximate value of $z$, required by Newton's method, can be computed assuming only one interferer nearest to the transmitter $i$ as discussed in the Appendix. The negative sign in (\ref{eq:area_integral_2}) is automatically negated by the dot product of vectors $(z(k)-z_{i})$ and $\nabla S_{i}(z(k))$.

\subsection{Capacity}

The expression to derive the capacity is
$$
c(z,\beta,\alpha)=\E(N(z,\beta,\alpha))=N(z,\beta,\alpha)=\lambda\sigma(\lambda,\beta,\alpha)~,
$$ 
where $\sigma(\lambda,\beta,\alpha)$ is computed using the above described method. 

\section{Medium Access Schemes Based on Exclusion Rules}
\label{sec:practical}

As we mentioned earlier, an accurate model of interference distribution, in case of exclusion schemes, is very hard to derive because of the correlation in the locations of simultaneous transmitters, {\it e.g.}, simultaneous transmitters separated by a certain minimum distance. We also mentioned a few approaches to model distribution of transmitters in exclusion schemes including Mat\'ern point process \cite{CSMA-Model,Weber3}  or, more recently, {\em SSI} point process \cite{Busson}.

Classical Mat\'ern point process is a location dependent thinning of PPP such that the remaining points are at least at a certain distance \mbox{$b>0$} from each other. Instead of distance,~\cite{CSMA-Model,Weber3} used received power level, a function of distance, in order to model CSMA based schemes. However, \cite{Busson} showed that Mat\'ern point process may lead to an underestimation of the density of simultaneous transmitters and proposed an {\em SSI} point process in which a node can transmit or, in other words, can be added to the set of simultaneous transmitters ${\cal S}$ if it is at least at a certain distance \mbox{$b>0$} from all active transmitters ({\it i.e.}, the transmitters already in the set ${\cal S}$). {\em SSI} point process can be used to study node coloring schemes but may result in a flawed representation of CSMA based schemes. In case of CSMA based schemes, a node senses the medium to detect if the signal level of interference is below a certain threshold and only transmits if it remains below that threshold during the randomly selected back-off period. Therefore, the {\em decision} of transmission depends on all nodes which are already active and transmitting on the medium. In order to address this inaccuracy in {\em SSI} point process model, \cite{Busson} proposed an {\em SSI$_k$} point process which ensures that a node, before transmitting on the medium, takes into account the interference from $k$ nearest active transmitters. However, very few analytical results are available on {\em SSI} and {\em SSI$_k$} point processes and most of the results are obtained via simulations. Therefore, we will use Monte Carlo simulation along with the analytical method proposed in \S \ref{sec:rx_area_2} to compute the capacity of node coloring and CSMA based schemes. 

In this section, we will discuss the models of node coloring and CSMA based schemes which we will employ in our Monte Carlo simulation later. In the following discussion, the set of all nodes in the network is {\cal N}. In practical implementation, this set is finite but in theory, it can be infinite but with a uniform density.

\subsection{Node Coloring Based Schemes}
\label{sec:tdma}

Node coloring schemes use a managed transmission scheme based on time division multiple access (TDMA) approach. The aim is to minimize the interference between transmissions that cause packet loss. These schemes assign colors to nodes that correspond to periodic slots, {\it i.e.}, nodes that satisfy a spatial condition, either based on physical distance or distance in terms of number of hops, will be assigned different colors. This condition is usually derived from the interference models of wireless networks such as unit disk graph (UDG) models. For example, in order to avoid collisions at receivers, all nodes within $k$ hops are assigned unique colors. Typical value of $k$  is $2$. A few practical implementations of node coloring schemes are as follows. In~\cite{unified}, authors proposed coloring based on RAND, MNF and PMNF algorithms. In RAND, nodes are colored in a random order whereas MNF and PMNF prioritize nodes according to the number of their neighbors. NAMA~\cite{NAMA} colors the nodes, in $2$-hop neighborhood, randomly using a hash number. In SEEDEX~\cite{SEEDEX}, nodes use random seed number in $2$-hop neighborhood to randomly elect the transmitter. FPRP~\cite{FPRP} is a five-phase protocol where nodes contend to allocate slots in $2$-hop neighborhood. DRAND is the distributed version of RAND and article~\cite{DRAND} shows its better performance as compared to SEEDEX and FPRP. \cite{tdma_fair} proposes a joint TDMA scheduling/load balancing algorithm for wireless multi-hop networks and shows that it can improve the performance in terms of multi-hop throughput and fairness. Most of these schemes use unit disk graph based interference model. However, success of a transmission depends on whether the SIR at the receiver is above a certain threshold. \cite{Derbel} is an example of a node coloring scheme which uses SIR based interference model. Note that, extremely managed transmission scheduling in node coloring schemes has significant overhead, {\it e.g.}, because of the control traffic or message passing required to achieve the distributed algorithms that resolve color assignment conflicts. 

In this article, instead of considering any particular scheme, we will present a model which ensures that transmitters use an exclusion distance in order to avoid the use of the same slot within a certain distance. This exclusion distance is defined in terms of euclidean distance $d$ which may be derived from the distance parameter of a typical TDMA-based scheme. Therefore, a slot cannot be shared within a distance of $d$ or, in other words, nodes transmitting in the same slot shall be located at a distance greater or equal to $d$ from each other. 

Following is a model of node coloring schemes which constructs the set of simultaneous transmitters, ${\cal S}$, in each slot (this is supposed to be done off-line so that transmission patterns periodically recur in each slot).
\begin{compactenum}
\item Initialize ${\cal M}={\cal N}$ and ${\cal S}=\emptyset$.
\item Randomly select a node $s_{i}$ from ${\cal M}$ and add it to the set ${\cal S}$, {\it i.e}, ${\cal S}={\cal S}\cup\{s_{i}\}$. Remove $s_{i}$ from the set ${\cal M}$.
\item Remove all nodes from the set ${\cal M}$ which are at distance less than $d$ from $s_{i}$.
\item If set ${\cal M}$ is non-empty, repeat from step $2$.
\end{compactenum}
Above described steps model a centralized or distributed node coloring scheme which {\em randomly} selects the nodes for coloring while satisfying the constraints of euclidean distance. Under given constraints, this model also tries to maximize the number of simultaneous transmitters in each slot and should give the maximum capacity achievable with any node coloring scheme which {\em may not} prioritize the nodes for coloring, {\it e.g.}, \cite{DRAND}.

\subsection{CSMA Based Schemes}
\label{sec:csma}

As compared to managed transmission schemes like node coloring, CSMA based schemes are simpler but are more demanding on the physical layer. Before transmitting on the channel, a node verifies if the medium is idle by sensing the signal level. If the detected signal level is below a certain threshold, medium is assumed idle and the node transmits its packet. Otherwise, it may invoke a random back-off mechanism and wait before attempting a retransmission. CSMA/CD (CSMA with collision detection) and CSMA/CA (CSMA with collision avoidance), which is also used in IEEE 802.11, are the modifications of CSMA for performance improvement. 

We will adopt a model of CSMA based scheme where nodes contend to access medium at the beginning of each slot. In other words, nodes transmit only after detecting that medium is idle. We assume that nodes defer their transmission by a tiny back-off time, from the beginning of a slot, and abort their transmission if they detect that medium is not idle. We also suppose that detection time and receive to transmit transition times are negligible and, in order to avoid collisions, nodes use randomly selected (but different) back-off times. Therefore, the main effect of back-off times is in the production of a random order of the nodes in competition.  

For the evaluation of the performance of CSMA based scheme, we will use the following simplified construction of the set of simultaneous transmitters ${\cal S}$. 
\begin{compactenum}
\item Initialize ${\cal M}={\cal N}$ and ${\cal S}=\emptyset$.
\item Randomly select a node $s_{i}$ from ${\cal M}$ and add it to the set ${\cal S}$, {\it i.e.}, ${\cal S}={\cal S}\cup\{s_{i}\}$. Remove $s_{i}$ from the set ${\cal M}$.
\item Remove all nodes from the set ${\cal M}$ which can detect a combined interference signal of power higher than $\theta$ (carrier sense threshold), from all transmitters in the set ${\cal S}$, {\it i.e.}, if 
$$
\sum_{s_i\in\cal{S}}|z_i-z_j|^{-\alpha}\geq \theta~, 
$$ 
remove $s_j$ from $\cal{M}$. Here, $z_i$ is the position of $s_i$ and $|z_i-z_j|$ is the euclidean distance between $s_i$ and $s_j$.
\item If set ${\cal M}$ is non-empty, repeat from step $2$.
\end{compactenum}
These steps model a CSMA based scheme which requires that transmitters do not detect an interference of signal level equal to or higher than $\theta$, during their back-off periods, before transmitting on the medium. At the end of the construction of set ${\cal S}$, some transmitters may experience interference of signal level higher than $\theta$. However, this behavior is in compliance with a realistic CSMA based scheme where nodes, which started their transmissions, or, in other words, are already added to the set ${\cal S}$ do not consider the increase in signal level of interference resulting from later transmitters.

\subsection{Reception Areas} 

We are not aware of an analytical closed-form expression for the probability distribution of signal level with node coloring or CSMA based schemes. Consequently, we do not have a closed form expression for the average size of the reception area of an arbitrary transmitter with these schemes. Therefore, it is evaluated via Monte Carlo simulation using the analytical method of \S \ref{sec:rx_area_2}. The value of $d$, in case of node coloring scheme, or $\theta$, in case of CSMA based scheme, can be tuned to obtain an average transmitter density of $\lambda$. 

\subsection{Capacity}

Similarly
$$
c(z,\beta,\alpha)=\E(N(z,\beta,\alpha))=\lambda \sigma(\lambda,\beta,\alpha)~,
$$ 
is also computed via Monte Carlo simulation. The capacity, $c(z,\beta,\alpha)$, is invariant for any homothetic transformation of $\lambda$ and, therefore, it is also independent of the values of medium access scheme parameters $\theta$ or $d$.

\section{Slotted ALOHA Scheme}
\label{sec:aloha}

In slotted ALOHA scheme, nodes do not use any complicated managed transmission scheduling and transmit their packets independently (with a certain medium access probability), {\it i.e.}, in each slot, each node decides independently whether to transmit or otherwise remain silent. Therefore, the set of simultaneous transmitters, in each slot, can be given by a uniform Poisson distribution of mean equal to $\lambda$ transmitters per unit square area~\cite{Jacquet:2009,SR-ALOHA,Weber2}. 

We can write \eqref{eq:sinr_condition} as
$$
|z-z_{i}|^{-\alpha}\geq \beta\underset{j\neq i}{\sum}|z-z_{j}|^{-\alpha}\label{eq:snr_relation_at_z}
$$
or 
$
{\cal W}(z,\{z_{i}\})\geq \beta {\cal W}(z,{\cal S}-\{z_{i}\})$, where ${\cal W}(z,{\cal S})=\underset{z_{j}\in {\cal S}}{\sum}|z-z_{j}|^{-\alpha}~.
$ 

\subsection{Distribution of Signal Levels}

\begin{figure}[!t]
\centering
\includegraphics[scale=0.65]{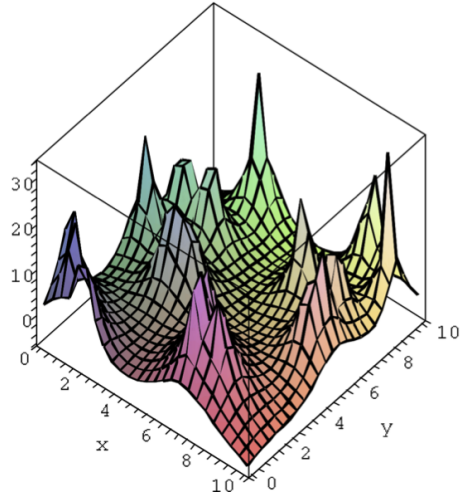}
\caption{Signal Levels (in dB's) for a random network with attenuation coefficient $\alpha=2.5$.
\label{fig:signal_levels}}
\end{figure}

Figure \ref{fig:signal_levels} shows the function ${\cal W}(z,{\cal S})$ for $z$ varying in the plane with ${\cal S}$ an arbitrary set of Poisson distributed transmitters. Figure \ref{fig:signal_levels} uses \mbox{$\alpha=2.5$}. It is clear that closer the receiver is to the transmitter, larger is the SIR. For each value of $\beta$ we can draw an area, around each transmitter, where its signal can be received with SIR greater or equal to $\beta$. Figure \ref{fig:reception_area} shows reception areas for the same set ${\cal S}$, as in Fig. \ref{fig:signal_levels}, for various values of $\beta$. As can be seen, the reception areas do not overlap for $\beta>1$ since there is only one dominant signal. For each value of $\beta$ we can draw, around each transmitter, the area where its signal is received with SIR greater or equal to $\beta$. Our aim is to find the average size of this area and how it is a function of $\lambda$, $\beta$ and $\alpha$.

${\cal W}(z,{\cal S})$ depends on ${\cal S}$ and hence is also a random variable. The random variable ${\cal W}(z,{\cal S})$ has a distribution which is invariant by translation and therefore does not depend on $z$. Therefore, we denote ${\cal W}(\lambda)\equiv{\cal W}(z,\lambda)$. Let $w({\cal S})$ be its density function. The set ${\cal S}$ is given by a $2D$ Poisson process with intensity $\lambda$ transmitters per slot per unit square area and Laplace transform of $w({\cal S})$, $\tilde{w}(\theta,\lambda)$, can be computed exactly. The Laplace transform, $\tilde{w}(\theta,\lambda)=\exp(\int(e^{-\theta r^{-\alpha}}-1)rdr)$, satisfies the identity
$$
\tilde{w}(\theta,\lambda)=\exp\left(-\pi\lambda\Gamma(1-\frac{2}{\alpha})\theta^{\frac{2}{\alpha}}\right)~,
$$
where $\Gamma(.)$ is the Gamma function. 

Note that, in all cases, $\tilde{w}(\theta,\lambda)$ is of the form $\exp(-\lambda C\theta^{\gamma})$ where $\gamma=\frac{2}{\alpha}$, and the expression of $C$ in case of $2D$ map is $C=\pi\Gamma(1-\gamma)$.

\begin{figure}[!t]
\centering
\includegraphics[clip,scale=0.65]{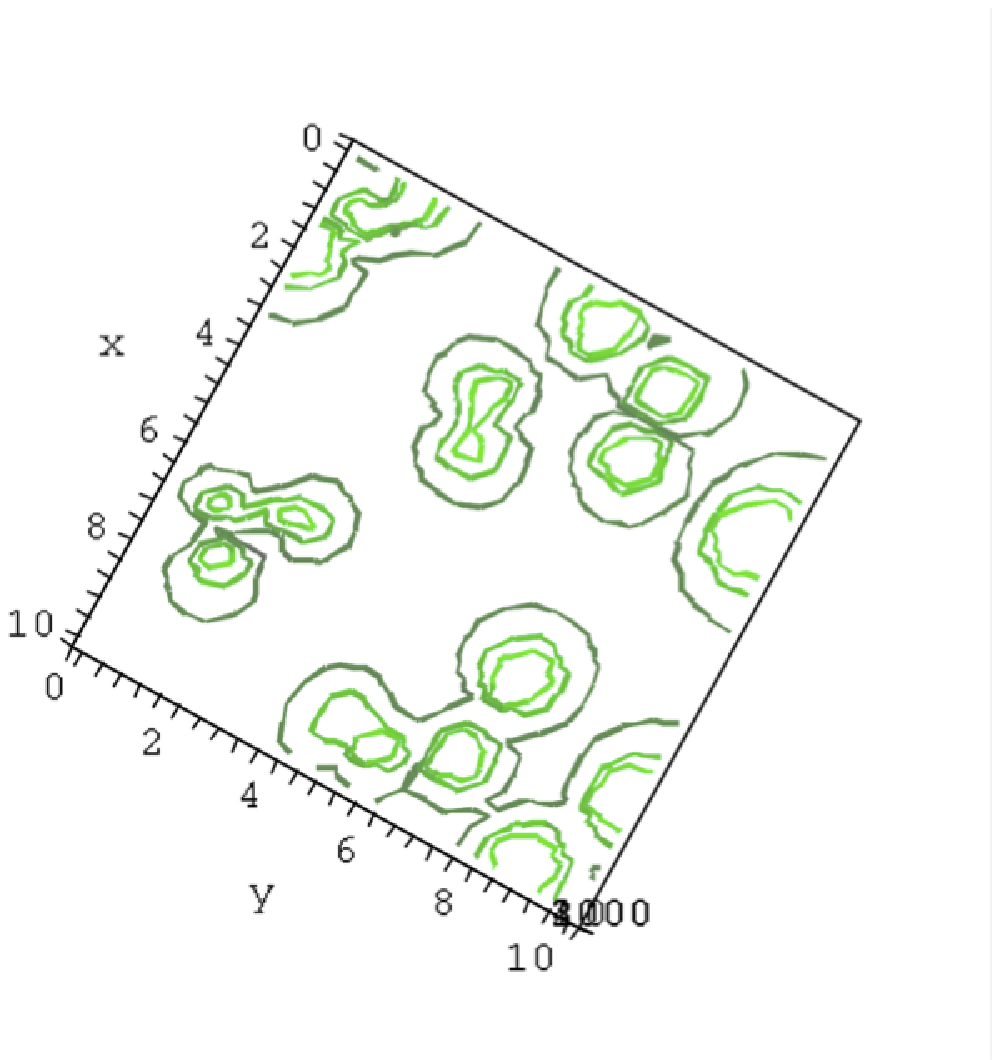}
\caption{Distribution of reception areas for various value of SIR. $\beta=1,4,10$ for situation of Fig. \ref{fig:signal_levels}.
\label{fig:reception_area}}
\end{figure}

From the above formula and by applying the inverse Laplace transformation, we get
$$
P({\cal W}(\lambda)<x)=\frac{1}{2i\pi}\underset{-i\infty}{\overset{+i\infty}{\int}}\frac{\tilde{w}(\theta,\lambda)}{\theta}e^{\theta x}d\theta
$$
Expanding $\tilde{w}(\theta,\lambda)=\underset{n\geq0}{\sum}\frac{(-C\lambda)^{n}}{n!}\theta^{n\gamma}$, we get
$$
P({\cal W}(\lambda)<x)=\frac{1}{2i\pi}\underset{n\geq0}{\sum}\frac{(-C\lambda)^{n}}{n!}\underset{-i\infty}{\overset{+i\infty}{\int}}\theta^{n\gamma-1}e^{\theta x}d\theta
$$
By bending the integration path towards negative axis
\begin{align*}
\frac{1}{2i\pi}\underset{-i\infty}{\overset{+i\infty}{\int}}\theta^{n\gamma-1}e^{\theta x}d\theta & =\frac{e^{i\pi n\gamma}-e^{-i\pi n\gamma}}{2i\pi}\underset{0}{\overset{\infty}{\int}}\theta^{n\gamma-1}e^{-\theta x}d\theta\\
 & =\frac{\sin(\pi n\gamma)}{\pi}\Gamma(n\gamma)x^{-n\gamma}\end{align*}
we get,
\begin{equation}
P({\cal W}(\lambda)<x)=\underset{n\geq0}{\sum}\frac{(-C\lambda)^{n}}{n!}\frac{\sin(\pi n\gamma)}{\pi}\Gamma(n\gamma)x^{-n\gamma}\label{eq:signal_pd}~.
\end{equation}

\subsection{Reception Areas}

Let $p(\lambda, r, \beta, \alpha)$ be the probability to receive a signal at distance $r$ with SIR at least equal to $\beta$. Therefore, we have 
$$
p(\lambda, r, \beta, \alpha)=P\left({\cal W}(\lambda) < \frac{r^{-\alpha}}{\beta}\right)
$$ 
and the average size of the reception area around an arbitrary transmitter with SIR at least equal to $\beta$ is 
$$
\sigma(\lambda, \beta, \alpha)=2\pi\int p(\lambda, r, \beta, \alpha)rdr~.
$$
Because of the obvious homothetic invariance, we can also write $p(\lambda,r,\beta,\alpha)=p(1,r\sqrt{\lambda},\beta,\alpha)$. Therefore, $\sigma(\lambda,\beta,\alpha)=\frac{1}{\lambda}\sigma(1,\beta,\alpha)$. Let $\sigma(1,\beta,\alpha)=2\pi\int P({\cal W}(1)<\frac{r^{-\alpha}}{\beta})rdr$. Assuming that $\tilde{w}(\theta,1)=\exp(-C\theta^\gamma)$ and using integration by parts, we get
$$
\sigma(1,\beta,\alpha)=\pi\alpha\int_0^{\infty}P\left({\cal W}(1)=\frac{1}{\beta r^{\alpha}}\right)\frac{r}{\beta r^\alpha}dr~.
$$
We use the fact that $P({\cal W}(1)=x)=\frac{1}{2\pi i}\int_{\cal C}w(\theta,1)e^{\theta x}d\theta$, where ${\cal C}$ is an integration path in the definition domain of $w(\theta,1)$, {\it i.e.}, parallel to the imaginary axis with $\Re(\theta)>0$. By changing the variable $x=(\beta r^{\alpha})^{-1}$ and inverting integrations, we get
\begin{align}
\sigma(1,\beta,\alpha)&=\frac{1}{2i}\int_{\cal C}w(\theta,1)\int_0^{\infty}e^{\theta x}(\beta x)^{-\gamma}dx \notag \\
&=\frac{1}{2i}\beta^{-\gamma}\Gamma(1-\gamma)\int_{\cal C}w(\theta,1)(-\theta)^{\gamma-1}d\theta~.\notag
\end{align}
Using $\tilde{w}(\theta,1)=\exp(-C\theta^{\gamma})$, and deforming the integration path to stick to the negative axis, we obtain
\begin{align}
\sigma(1,\beta,\alpha)&=\frac{e^{i\pi\gamma}-e^{-i\pi\gamma}}{2i}\beta^{-\gamma}\Gamma(1-\gamma)\int_0^{\infty}\exp(-C\theta^{\gamma})\theta^{\gamma-1}d\theta \notag \\
&=\sin(\pi \gamma)\beta^{-\gamma}\frac{\Gamma(1-\gamma)}{C\gamma}~.\notag
\end{align}
Using the expression for $C$, we have 
$$
\sigma(1,\beta,\alpha)=\frac{\sin(\pi\gamma)}{\pi\gamma}\beta^{-\gamma}~.
$$
Therefore, the average size of the reception area around an arbitrary transmitter $i$ with SIR at least equal to $\beta$ satisfies the identity
\begin{equation}
\sigma_{i}(\lambda,\beta,\alpha)=\frac{1}{\lambda}\frac{\sin(\frac{2}{\alpha}\pi)}{\frac{2}{\alpha}\pi}\beta^{-\frac{2}{\alpha}}~.
\label{eq:poisson_area}
\end{equation}

The reception area $\sigma_{i}(\lambda,\beta,\alpha)$ is inversely proportional to the density of transmitters $\lambda$ and the product $\lambda\sigma_{i}(\lambda, \beta, \alpha)$ is a function of $\beta$ and $\alpha$. We notice that when $\alpha$ approaches infinity, $\sigma_{i}(\lambda,\beta,\infty)$ approaches $1/\lambda$. This is due to the fact that when $\alpha$ is very large, all nodes other than the closest transmitter tend to contribute as a negligible source of interference and consequently the reception areas turn to be the Voronoi cells around every transmitter. This holds for all values of $\beta$. The average size of Voronoi cell being equal to the inverse density of the transmitters, $\frac{1}{\lambda}$, we get the asymptotic result. Note that when $\beta$ grows as $\exp(O(\alpha))$, we have $\sigma_{i}(\lambda,\beta,\alpha)\approx\frac{1}{\lambda}\exp(-\frac{2}{\alpha}\log(\beta))$, which suggests that the typical SIR as $\alpha\rightarrow\infty$ is of the order of $\exp(O(\alpha))$.

Secondly, as $\alpha$ approaches $2$, $\sigma_{i}(\lambda,\beta,2)$ approaches zero because $\sin(\frac{2}{\alpha}\pi)$ approaches zero. Indeed, the contribution of remote nodes tends to diverge and makes the SIR approach to zero. This explains why $\sigma_{i}(\lambda,\beta,2)$ approaches to zero for any fixed value of $\beta$.

\subsection{Capacity}

In this case, the analytical expressions (\ref{eq:poisson_hand_over_no}) and (\ref{eq:poisson_area}) lead to 
\begin{equation}
c(z,\beta,\alpha)=\E(N(z,\beta,\lambda))=\sigma(1,\beta,\alpha)~.
\label{eq:poisson_capacity}
\end{equation}

\section{Evaluation and Results}
\label{sec:simulations}

In order to approach an infinite map, we perform numerical simulations in a very large network spread over $2D$ square map with length of each side equal to $10000$ meters.

\subsection{Grid Pattern Based Schemes}

In case of grid pattern schemes, transmitters are spread over this network area in square, hexagonal or triangular pattern. For all grid patterns, we set $d$ equal to $25$ meters although it will have no effect on the validity of our conclusions as capacity, $c(z,\beta,\alpha)$, is independent of $\lambda$. To keep away edge effects, we compute the size of the reception area of transmitter $i$, located in the center of the network area: \mbox{$z_{i}=(x_{i},y_{i})=(0,0)$}. The network area is large enough so that the reception area of transmitter $i$ is close to its reception area in an infinite map. $\lambda$ depends on the type of grid and it is computed from the total number of transmitters spreading over the network area of $10000\times10000$ square meters. In our numerical simulations, we set $\delta t=0.01$.

\subsection{Medium Access Schemes Based on Exclusion Rules}

\subsubsection{Node Coloring Based Schemes}

Performance of node coloring based schemes is analyzed, via simulations, using the model specified in \S \ref{sec:tdma}. We set $d$ equal to $25$ meters. 

\subsubsection{CSMA Based Schemes}

In order to evaluate the capacity of CSMA based schemes, we perform simulations using the model specified in \S \ref{sec:csma}. The value of carrier sense threshold, $\theta$, is set equal to $1\times 10^{-5}$.

\subsubsection{Simulations}

We consider that nodes are uniformly distributed over the network area of $10000\times10000$ square meters. Simultaneous transmitters, in each slot, are selected according to the model of each medium access scheme. Considering the practical limitations introduced by the bounded network area, we use the following Monte Carlo method to evaluate $\sigma(\lambda,\beta,\alpha)$. We {\em only} compute the size of the reception area of a transmitter located nearest to the center of the network area and $\sigma(\lambda,\beta,\alpha)$ is the average of results obtained with $10000$ samples of node distributions. Similarly, $\lambda$ is also the average of the density of simultaneous transmitters obtained with these $10000$ samples of node distributions. Note that the models of medium access schemes select simultaneous transmitters randomly and transmitters are uniformly distributed over the network area. Therefore, using Monte Carlo method, {\it i.e.}, a large number of samples of node distributions and, with each sample, only measuring the reception area of a transmitter located nearest to the center of the network area gives an accurate approximation of $\sigma(\lambda,\beta,\alpha)$ in an infinite map with given values of $d$ or $\theta$. 

It can be argued that, in case of CSMA based scheme, density of simultaneous transmitters is higher on the boundaries of the network area, because of lower signal level of interference, as compared to the central region. The network area is very large and we observed that the difference, in spatial density of simultaneous transmitters, on the boundaries and central region is negligible. We also know that the capacity, $c(z,\beta,\alpha)$, is independent of $\lambda$ which depends on $d$ or $\theta$. However, if the node density is very low, it will also have an impact on the packing (density) of simultaneous transmitters in the network. Therefore, $\lambda$ should be maximized to the point where no additional transmitter can be added to the network under given values of $d$ or $\theta$. This can be achieved by keeping the node density very high, {\it e.g.}, we observed that the node density of $1$ node per square meter is sufficient and further increasing the node density does not increase $\lambda$. In order to keep away the edge effects, values of $d$ or $\theta$ are chosen such that $\lambda$ is sufficiently high and edge effects have minimal effect on the central region of the network.

\subsection{Slotted ALOHA Scheme}

In case of slotted ALOHA scheme, capacity, $c(z,\beta,\alpha)$, is computed from analytic expressions (\ref{eq:poisson_area}) and (\ref{eq:poisson_capacity}).

\begin{figure*}[!t]
\centering
\subfloat[$\beta$ is varying and $\alpha$ is fixed at $4.0$.]{
	\hspace{-0.85 cm}
	\includegraphics[scale=1]{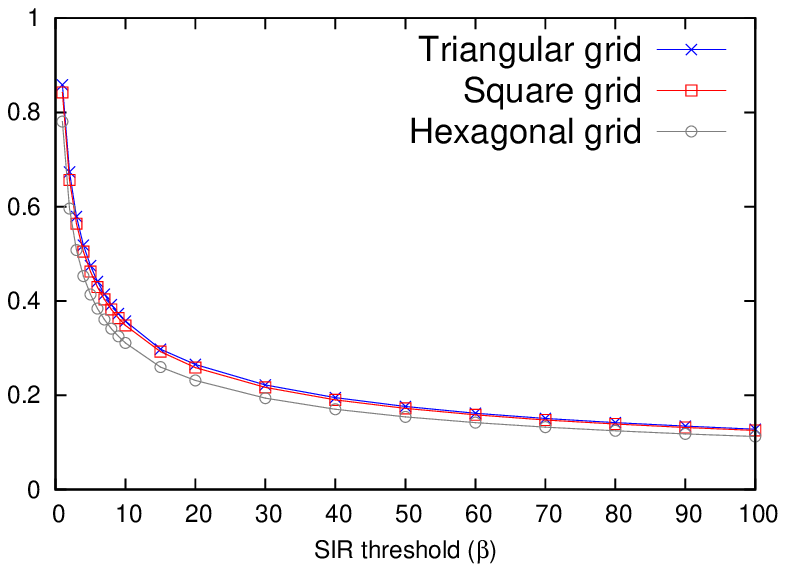}
	\hspace{-0.85 cm}
	\includegraphics[scale=1]{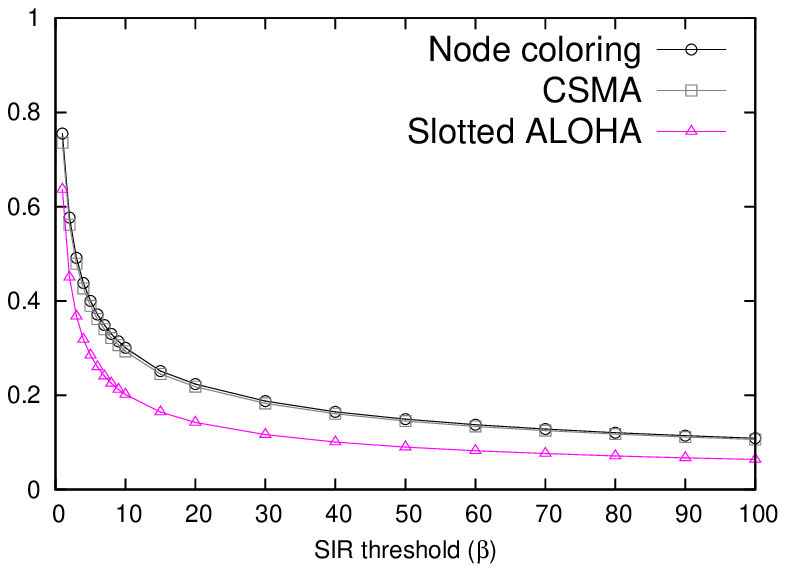}
}

\subfloat[$\beta$ is fixed at $10.0$ and $\alpha$ is varying.]{
	\hspace{-0.85 cm}
	\includegraphics[scale=1]{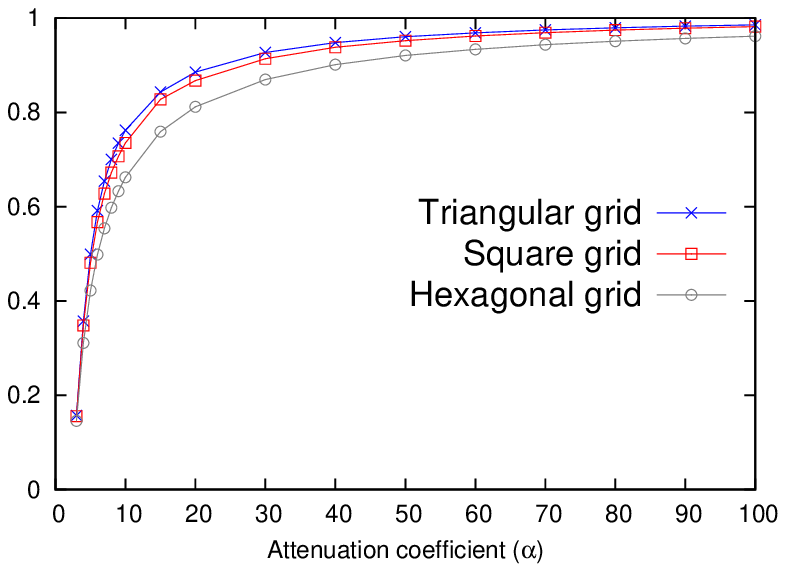}
	\hspace{-0.85 cm}
	\includegraphics[scale=1]{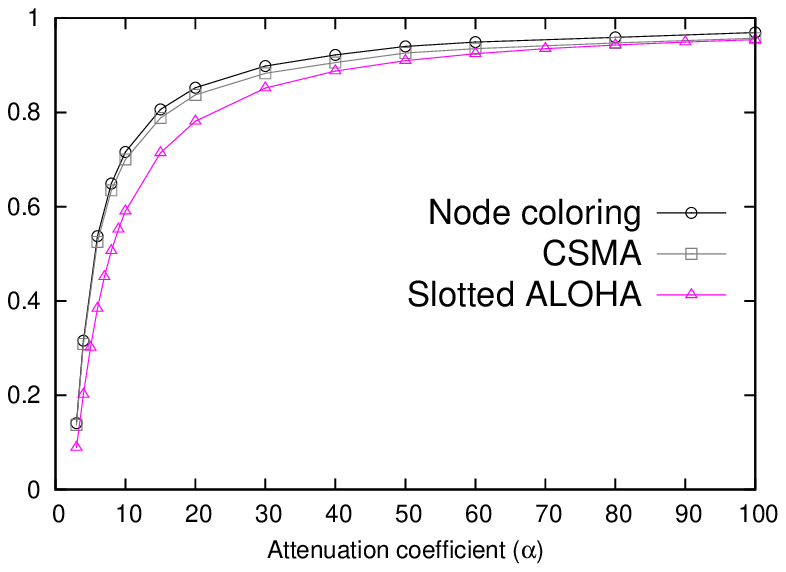}
}
\caption{Capacity, $c(z,\beta,\alpha)$, of grid pattern (triangular, square and hexagonal), node coloring, CSMA and slotted ALOHA based medium access schemes.
\label{fig:comparison}}
\end{figure*}

\subsection{Observations}

The values of SIR threshold, $\beta$, and attenuation coefficient, $\alpha$, depend on the underlying physical layer or system parameters and are usually fixed and beyond the control of network/protocol designers. However, to give the reader an understanding of the influence of these parameters on the capacity, $c(z,\beta,\alpha)$, of different medium access schemes, we assume that these parameters are variable. 

Figure \ref{fig:comparison}(a) shows the comparison of capacity, $c(z,\beta,\alpha)$, with grid patterns, node coloring, CSMA and slotted ALOHA schemes with $\beta$ varying and \mbox{$\alpha=4.0$}. Similarly, Fig. \ref{fig:comparison}(b) shows the comparison of these medium access schemes with \mbox{$\beta=10.0$} and $\alpha$ varying. We know that as $\alpha$ approaches infinity, reception area around each transmitter turns to be a Voronoi cell with an average size equal to $1/\lambda$. Therefore, as $\alpha$ approaches infinity, $c(z,\beta,\alpha)$ approaches 1. For slotted ALOHA scheme, (\ref{eq:poisson_area}) and (\ref{eq:poisson_capacity}) also arrive at the same result. For other medium access schemes, we computed $c(z,\beta,\alpha)$ with $\alpha$ increasing up to $100$ and from the results, we can observe that asymptotically, as $\alpha$ approaches infinity, $c(z,\beta,\alpha)$ approaching 1 is true for all schemes. 

\subsubsection{Optimal Local Capacity in Wireless Ad Hoc Networks}

From the results, we can see that the maximum capacity in wireless ad hoc networks can be obtained with triangular grid pattern based medium access scheme. In order to quantify the improvement in capacity by triangular grid pattern scheme over other schemes, we perform a scaled comparison of triangular grid pattern, slotted ALOHA, node coloring and CSMA based schemes which is obtained by dividing the capacity, $c(z,\beta,\alpha)$, of all these schemes with the capacity, $c(z,\beta,\alpha)$, of triangular grid pattern scheme. Figure \ref{fig:improvement}  shows the scaled comparison with $\beta$ and $\alpha$ varying. It can be observed that triangular grid pattern scheme can achieve, {\em at most}, double the capacity of slotted ALOHA scheme. However, node coloring and CSMA based medium access schemes can achieve almost \mbox{$85\sim90\%$} of the optimal capacity obtained with triangular grid pattern scheme. 

\subsubsection{Observations on Node Coloring Scheme}

Triangular grid pattern based medium access scheme can be visualized as an optimal node coloring which ensures that transmitters are exactly at distance $d$ from each other whereas, in case of random node coloring, transmitters are selected randomly and only condition is that they must be at a distance greater or equal to $d$ from each other. The exclusion region around each transmitter is a circular disk of radius $d/2$ with transmitter at the center. Note that the disks of simultaneous transmitters shall not overlap. The triangular grid pattern can achieve a packing density of \mbox{$\pi/\sqrt{12}\approx0.9069$}. The packing density is defined as the proportion of network area covered by the disks of simultaneous transmitters. On the other hand, random packing of disks, which is the case in random node coloring, can achieve a packing density in the range of \mbox{$0.54\sim0.56$} only~\cite{disk,Busson}. We have seen in the results that even this sub-optimal packing of simultaneous transmitters by random node coloring scheme can achieve almost similar capacity as obtained with optimal packing by triangular grid pattern based medium access scheme. 

\subsubsection{Observations on CSMA Based Scheme}

We observe that the capacity with CSMA based scheme is slightly lower (by approximately $3\%$) as compared to node coloring based medium access scheme and this is irrespective of the value of carrier sense threshold. The reason of slightly lower capacity with CSMA based scheme is that the exclusion rule is based on carrier sense threshold, rather than the distance in-between simultaneous transmitters. With carrier sense baed exclusion rule, CSMA based scheme may not allow to pack more transmitters, in each slot, that would have been possible with node coloring schemes. In other words, CSMA may result in a lower packing density of simultaneous transmitters as compared to node coloring scheme. This can also be observed by comparing the densities of {\em SSI} and {\em SSI$_k$} point processes in \cite{Busson} and also explains the slightly lower capacity of CSMA based scheme as compared to node coloring scheme. However, as $\alpha$ approaches infinity, $\lambda$ with CSMA based scheme approaches the node density and reception area around each transmitter also becomes a Voronoi cell with an average size equal to the inverse of node density. In fact, asymptotically, as $\alpha$ approaches infinity, capacity, $c(z,K,\alpha)$, approaches $1$.

\begin{figure*}[!t]
\centering
\subfloat[$K$ is varying and $\alpha$ is fixed at $4.0$.]{
	\hspace{-0.85 cm}
	\includegraphics[scale=1]{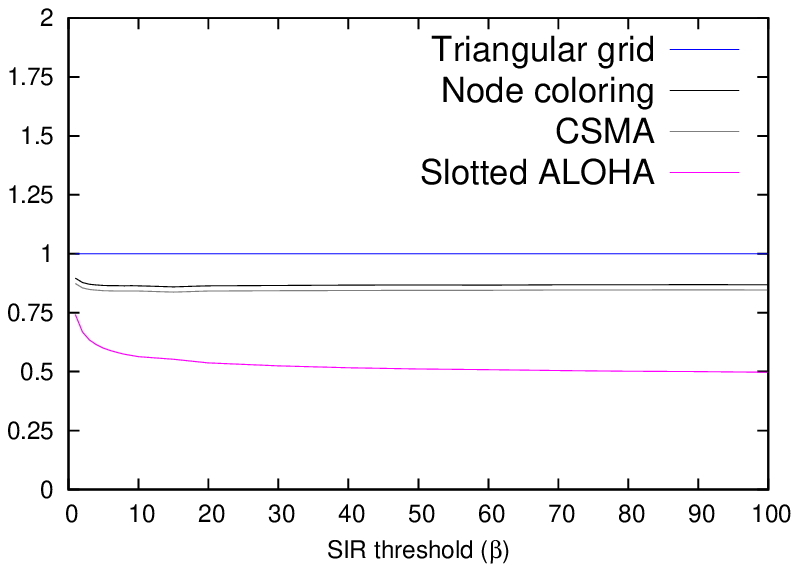}
}
\subfloat[$K$ is fixed at $10.0$ and $\alpha$ is varying.]{
	\hspace{-0.85 cm}
	\includegraphics[scale=1]{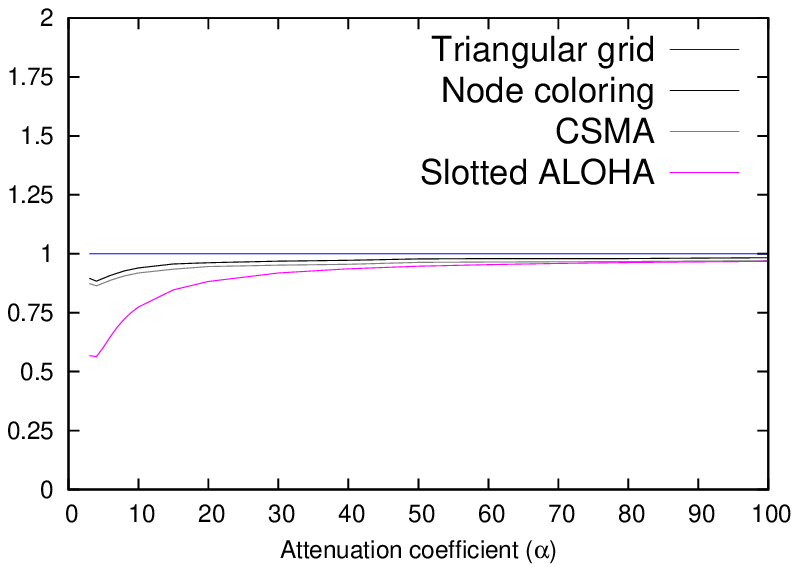}
}
\caption{Scaled comparison of triangular grid pattern, node coloring, CSMA and slotted ALOHA based schemes. 
\label{fig:improvement}}
\end{figure*}

\section{Current Limitations and Future Extensions}
\label{sec:future}

In future, we will extend this work to multi-hop networks. A medium access scheme which achieves higher local capacity should also be able to achieve higher end-to-end capacity in multi-hop networks. For example, consider that $\lambda$ is normalized across all medium access schemes to $1$. Therefore, higher local capacity means higher $\sigma(1,\beta,\alpha)$ which has an impact on the range of transmission and the number of hops required to reach the destination. The analysis to establish exact bounds on end-to-end capacity with different medium access schemes in multi-hop networks will be challenging as we will have to take into account the impact of routing schemes on capacity as well as various parameters like hop length, number of hops and density of simultaneous transmitters which are interrelated. 

The analysis presented here do not take into account fading and shadowing effects. Some results with fading are available, {\it e.g.}, for Poisson distribution of transmitters~\cite{Weber2,Bartek,Haenggi}. Our analysis, in case of slotted ALOHA, can take into account fading by using the results of~\cite{Jacquet:2009}. Nevertheless, analysis of all medium access schemes, discussed here, under the common framework, such as local capacity, is lacking. 

\section{Conclusions}
\label{sec:conclude}

We evaluated the performance of wireless ad hoc networks under the framework of local capacity. We used analytical tools, based on realistic interference model, to evaluate the performance of slotted ALOHA and grid pattern based medium access schemes and we used Monte Carlo simulation to evaluate node coloring and CSMA based schemes. Our analysis implies that maximum local capacity in wireless ad hoc networks can be achieved with grid pattern based schemes and our results show that triangular grid pattern outperforms square and hexagonal grids. Moreover, compared to slotted ALOHA, which does not use any significant protocol overhead, triangular grid pattern can only increase the capacity by a factor of $2$ or less whereas CSMA and node coloring can achieve almost similar capacity as the triangular grid pattern based medium access scheme.  

The conclusion of this work is that improvements above ALOHA are limited in performance and may be costly in terms of protocol overheads and that CSMA or node coloring can be very good candidates. Therefore, attention should be focused on optimizing existing medium access schemes and designing efficient routing strategies in case of multi-hop networks. Note that, our results are also relevant when nodes move according to an i.i.d. mobility process such that, at any time, distribution of nodes in the network is homogeneous.

\appendix

\subsection{Locating the Starting Point $z$ on the Closed Curve Bounding the Reception Area}

The SIR $S_{i}(z)$ at point $z$ should be greater or at least equal to $\beta$. We assume that it is equal to $\beta$. Therefore, at point $z$,  we have
$$
\frac{|z-z_{i}|^{-\alpha}}{\underset{j\neq i}{\sum}|z-z_{j}|^{-\alpha}}=\beta~.
$$
Our aim is to find the coordinates of point $z$ which satisfy the above relation. To simplify computation of point $z$ on the closed curve, bounding reception area of transmitter $i$ located at $z_{i}=(x_{i},y_{i})$, its $y$ coordinate can be fixed such that $z=(x,y)=(x,y_{i})$. This reduces the above equation to
\begin{equation}
\frac{|x-x_{i}|}{\underset{j\neq i}{\sum}|z-z_{j}|}-\beta=0~.
\label{eq:snr_newton_eq}
\end{equation}
Equation \eqref{eq:snr_newton_eq} is a function of variable $x$ and can be solved using Newton's Method.
\begin{quote}
\textit{Remark: Newton's Method: Given a function $f(x)$ and its derivative $f'(x)$, begin with a first guess $x_{0}$. Provided the function is reasonably well-behaved, a better approximation $x_{1}$ is $:=x_{0}-\frac{f(x_{0})}{f'(x_{0})}$. The process is repeated until a sufficiently accurate value is reached:}
\end{quote}
\begin{equation}
x_{n+1}=x_{n}-\frac{f(x_{n})}{f'(x_{n})}~.
\label{eq:newton_method}
\end{equation}

\begin{figure*}
\centering
\psfrag{a}{$z_i$}
\psfrag{b}{$z$}
\psfrag{c}{$z_j$}
\psfrag{d}{$r_i$}
\psfrag{e}{$r_j$}
\psfrag{f}{$d$}
\subfloat[Square and Triangular grid based schemes.]{
\centering
\includegraphics[scale=0.8]{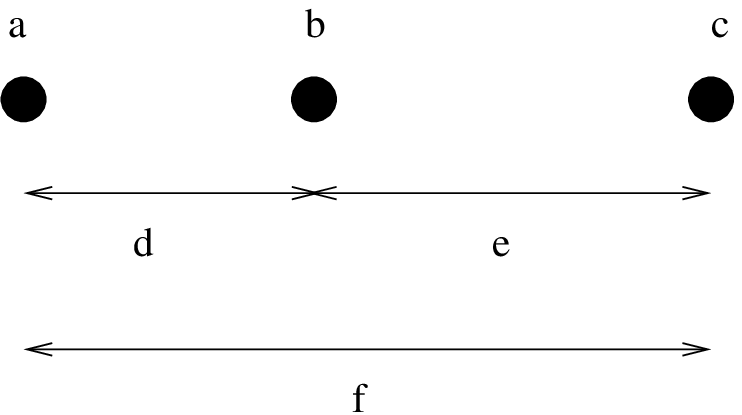}
}
\subfloat[Hexagonal grid pattern based scheme.]{
\centering
\includegraphics[scale=0.8]{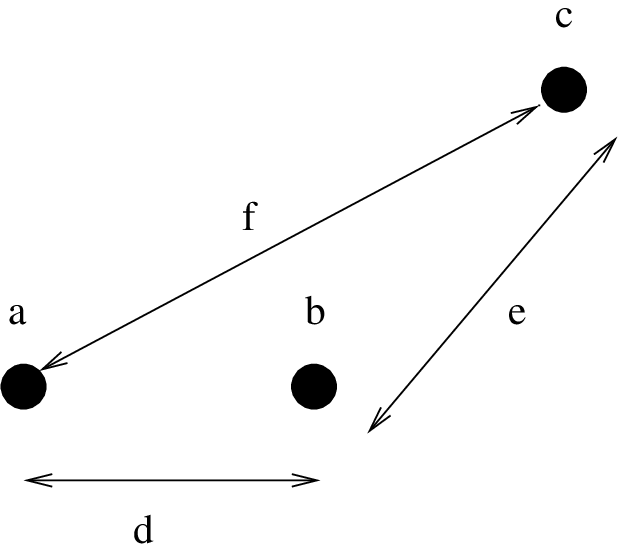}
}
\caption{Geometric representation of finding the first approximation of point $z=(x_{0},y_{i})$.
\label{fig:find_x0}}
\end{figure*}

From \eqref{eq:snr_newton_eq}
$$
f(x)=\frac{|x-x_{i}|^{-\alpha}}{\underset{j\neq i}{\sum}|z-z_{j}|^{-\alpha}}-\beta~,
$$
We assume that $g=r_{i}^{-\alpha}$ and $h=\underset{j\neq i}{\sum}r_{j}^{-\alpha}$ where $r_{i}=|x-x_{i}|$ and $r_{j}=|z-z_{j}|$.

Newton's Method requires the derivative of $f(x)$ which is computed as below.
$$
\frac{d}{dx}g=-\alpha\frac{x-x_{i}}{r_{i}^{\alpha+2}}
$$
$$
\frac{d}{dx}h=\underset{j\neq i}{\sum}-\alpha\frac{z-z_{j}}{r_{j}^{\alpha+2}}
$$
The first derivative of the function $f(x)$ is,
\begin{align*}
f'(x) & =\frac{1}{h^{2}}(h\frac{d}{dx}g-g\frac{d}{dx}h)\\
 & =\frac{\left((\underset{j\neq i}{\sum}r_{j}^{-\alpha})(-\alpha\frac{x-x_{i}}{r_{i}^{\alpha+2}})-(r_{i}^{-\alpha})(-\alpha\frac{z-z_{j}}{r_{j}^{\alpha+2}})\right)}{\frac{1}{\left[\underset{j\neq i}{\sum}r_{j}^{-\alpha}\right]}}
 \end{align*}

Newton's Method also requires first approximation of the root, $x_{0}$. An approximate value, closer to the actual root, can significantly reduce the number of iterations in Newton's Method. 

In all three types of grid networks, the transmitter closest to $i$, hereafter referred to as $j$, lies at distance $d$ and hence can give the best estimate $x_{0}$. For first approximation $x_{0}$, we can ignore all other transmitters in the network. In this case,

\begin{align*}
\frac{|z-z_{i}|^{-\alpha}}{|z-z_{j}|^{-\alpha}} & \geq \beta\\
\frac{r_{i}^{-\alpha}}{r_{j}^{-\alpha}} & \geq \beta:i\neq j~,
\end{align*}
where $r_{i}=|z-z_{i}|$ and $r_{j}=|z-z_{j}|$ and we get
$$
\frac{r_{j}}{r_{i}}\leq(\beta)^{\frac{1}{\alpha}}
$$

The location of transmitters $i$ and $j$ and point $z$ in the plane form three corners of a triangle with angle $\theta$ equal to $0$ in case of square and triangular grid and $\pi/6$ radians in case of hexagonal grid layout. Figure \ref{fig:find_x0} shows the location of transmitters $z_{i}$ and $z_{j}$, point $z$ and distances $r_{i}$, $r_{j}$ and $d$. Using above relation between $r_{i}$ and $r_{j}$ and the Law of Cosines, we get the solution of $r_{i}$ as

$$
r_{i}=\frac{-B\pm\sqrt{B^{2}-4AC}}{2A}~,
$$ 

where $A=1-K^{\frac{2}{\alpha}}$, $B=-2.d.\cos(\theta)$ and  $C=d^{2}$ where $d$ is the distance between transmitters $i$ and $j$ and is a known parameter of the grid layout. 
\begin{quote}
\textit{Remark: Select positive value of $r_{i}$ as the solution of the above quadratic equation. Using $x_{0}=x_{i}+r_{i}$ as the first approximate solution in Newton's Method \eqref{eq:newton_method}, and after a few iterations, we can get a sufficiently accurate value $x_{n+1}$ which will be the $x$ coordinate of the point $z$. }

\textit{The coordinates of point $z$ will be: $(x_{n+1},y_{i})$.}
\end{quote}

\bibliographystyle{hieeetr}
\bibliography{local_capacity_arxiv}

\end{document}